\documentclass[12pt,a4paper]{article}
\usepackage{amssymb,amsmath,amsthm}
\usepackage{epsfig}
\usepackage{tikz}
\usepackage{setspace}
\newtheorem{definition}{Definition}[section]
\newtheorem{proposition}[definition]{Proposition}

\newtheorem{example}[definition]{Example}
\newtheorem{remark}[definition]{Remark}

\newcommand{\RR}{\mathbb{R}}
\newcommand{\NN}{\mathbb{N}}
\newcommand{\BB}{\mathbb{B}}
\DeclareMathOperator*{\argmax}{arg\,max}

\newcommand{\Img}{\operatorname{Im}}
\newcommand{\fix}{\operatorname{fix}}

\newcommand{\power}[1]{{\mathcal P}(#1)}

\title{Higher-Order Decision Theory\footnote{\scriptsize We thank the seminar participants at the University of Mannheim, the Dagstuhl Perspectives Workshop ``Categorical Methods at the Crossroads'', the Dagstuhl Seminar ``Coalgebraic Semantics of Reflexive Economics'', the ``Computing in Economics and Finance Conference 2014'' in Oslo, the ``Cogrow'' Workshop in Nijmegen 2014 and the ``Logics for Social Behavior'' in Den Haag 2014 for helpful comments. Hedges thanks EPSRC, grant EP/K50290X/1, for financial support. Oliva gratefully acknowledegs financial support by the Royal Society through grant 516002.K501/RH/kk. Winschel and Zahn gratefully acknowledge financial support by the Deutsche Forschungsgemeinschaft (DFG) through SFB 884 ``Political Economy of Reforms''.}}
 
\author{
Jules Hedges, Paulo Oliva\\ 
\footnotesize School of Electronic Engineering and Computer Science, Queen Mary University London
\vspace{0.01cm}\\
Evguenia Sprits, Philipp Zahn\\ 
\footnotesize Department of Economics, University of Mannheim
\vspace{0.01cm}\\
Viktor Winschel\\ 
\footnotesize Department of Management, Technology and Economics, ETH Z\"urich
}
\date{\today}
\begin{document}
\maketitle
\begin{abstract}
\footnotesize
Classical decision theory models behaviour in terms of utility maximisation where utilities represent rational preference relations over outcomes.
However, empirical evidence and theoretical considerations suggest that we need to go beyond this framework.
We propose to represent goals by higher-order functions or operators that take other functions as arguments where the max and $\argmax$ operators are special cases.
Our higher-order functions take a context function as their argument where a context represents a process from actions to outcomes.
By that we can define goals being dependent on the actions and the process in addition to outcomes only.
This formulation generalises outcome based preferences to context-dependent goals.
We  show how to uniformly represent within our higher-order framework
classical utility maximisation but also various other extensions that have been debated in economics.
\end{abstract}

\section{Introduction}

Rational choice theory provides an elegant and succinct framework in order to model individual choices. The representation of agents' goals and their choices as utility maximisation has created a rich theory that is flexible and has vast applications. Despite its importance, rational choice theory does have its limitations. These limitations are conceptual as well as empirical and they are well documented in the literature.

The standard mathematical representation of goals as rational preference relations\footnote{Typically, a preference relation is defined to be \emph{rational}, if it is total and transitive, see \cite{Mas-Colell1995,Rubinstein2006,kreps2012microeconomic}} or alternatively as utility functions, can be too restrictive -- useful for many applications but not for all. Several issues have been identified. 
Take, for instance, the standard assumption that individuals' preference relations are complete. There has been a long debate as to whether this is a reasonable requirement, and several solutions have been proposed \cite{Aumann1962, Dubra_et_al2004_comple, Ok2002, Richter1971, Sen1997}.

Furthermore, there are numerous examples of decision procedures where the agents either do not not fully maximise -- like the satisficing behaviour of Herbert Simon -- or
they may adhere to various decision heuristics which violate for example the assumption of independence of irrelevant alternatives, like menu dependent  
or second best decision procedures \cite{Kalai2002,Sen1997}. 

A fundamental, often implicit, assumption is that people only care about the final outcomes of their decisions. Conceptually, it is not obvious why rational decision 	makers should be concerned with consequences of their actions only, without considering the action itself \cite{Sen1997}. 

Therefore, given the conceptual and empirical deficiencies, there are good reasons to not limit decision theory to the rational choice paradigm. 
Many of the proposed solutions, however, only attempt to represent seemingly `non-rational' behaviours as rational ones. This includes the creation of an extended outcome space, where aspects of the decision process are explicitly represented as outcomes \cite{Sen1997}. Alternatively, several papers propose to rationalise choice by multiple rationales \cite{Kalai2002} or to represent the preferences by multivalued utility functions \cite{Ok2002,Dubra_et_al2004_comple}. In all cases, the decision problem is usually artificially manipulated by the economist in order to represent the agent ``as if" he would decide given some rational preference relation or maximise a utility function, even if it is clear from the description of the decision process that actually nothing like that takes place. 

In this paper we provide a new approach based on higher-order functions
that unifies the behavioural patterns mentioned above with the classical approach based on  rational behaviour and utility maximisation. 
A \emph{higher-order function} (or \emph{functional}) is a function whose domain is itself a set of functions.
This paper relies on a framework that has been developed in computer science \cite{escardo10a,escardo_sequential_2011} as a game theoretical approach to proof theory\footnote{Proof theory is a branch of mathematical logic which investigates the structure and meaning of formal mathematical proofs. It has been recently discovered that certain proofs of high logical complexity can be interpreted as computer programs which compute equilibria of suitable generalised games.}. 
We extend and apply this approach to decision theory.\footnote{In the companion paper \cite{Hedges_et_al_2015_games}, 
we extend the higher-order framework to game theory.} 

The core concept is that we model agents' goals as \emph{quantifiers}, 
i.e. higher-order functions of type $(X \to R)\to R$, where $X\to R\,$ is a space of functions from the set of choices $X$ to the set of possible outcomes $R$. Quantifiers describe which outcomes 
an agent considers to be good.  A corresponding notion is that of a \emph{selection function}, i.e. a higher-order function of type $(X \to R) \to X$ which calculates a choice that meets the desired goal. 
We take functions of type $X \to R$ to represent a \emph{decision context} such that formulating goals as quantifiers boils down to describing the preferred outcomes for any given decision context. 
The decision of an agent described by a quantifier takes place in a context
and by that the agent can take into account the process from actions to outcomes.

Since the $\max \colon (X \to \RR) \to \RR$ operator is a quantifier and  the corresponding $\argmax \colon (X \to \RR) \to X$ operator is a selection function,
we can show how to instantiate our higher-order approach with these operators
such that the usual utility functions and preference relations are instances of our modelling framework.

We also show how our framework captures alternatives to the usual rationality assumption.  More specifically, our framework addresses the relevance of the choice act itself since quantifiers and selection functions take the decision context as input such that context-dependent goals can be seamlessly modelled. 
The outcome space in our formulation of goals can have any arbitrary structure and is not restricted in order to be representable by utility functions (or, equivalently, rational preferences). In particular, we can model preference relations that are incomplete.
Moreover, our framework allows to model arbitrary heuristics (e.g. second best choice, median procedure, etc.) directly, without the need for a multivariable representation. 
Finally, we show that quantifiers themselves can be used to model behaviour directly at the level of higher-order functions. 
We will show how to represent the abstract concerns of coordination and differentiation as fixpoint operators that are higher-order functions as well.
In other words, our framework is able to represent existing models but it can also be used to formulate new ones. 

We do not aim to provide an answer to the question of which paradigm should be used in order to model a particular phenomenon. Instead, we introduce our framework such that different approaches can be captured as instances of the same abstraction. Our contribution is to provide a  powerful organising framework and a very expressive mathematical language. 
In this framework we can clarify the relation between different models, highlight their commonalities and differences and allow for their combinations.

The paper is organized as follows: We introduce and explain higher-order functions in the next section. Then, we instantiate rational preferences as well as utility maximisation as special cases in Section \ref{sec:decisions}. In Sections \ref{sec:beyond} and \ref{sec:reflexive} we introduce a series of deviations from classical choice theory and show how they can be 
represented by higher-order functions.  We conclude in Section \ref{sec:conclusions}.

\section{Agents as Quantifiers}
\label{sec:quantifiers}

A \emph{higher-order function} (or \emph{functional}) is a function whose domain is itself a set of functions. Given sets $X$ and $Y$ we denote by $X \to Y$ the set of all functions with domain $X$ and codomain $Y$. A higher-order function is therefore a function $f : (X \to Y) \to Z$ where $X$, $Y$ and $Z$ are sets.

There are examples of higher-order functions that are well familiar to economists. 
In case of the maximisation of a utility function $u \colon X \to \RR$
\[ \max_{x \in X} u(x) \]
the max operator takes the utility function $u \colon X \to \RR$ as its input and returns a real number $\max_{x \in X} u(x)$ as the output.
Thereore, the $\max$ operator has type
\[ \max \colon (X \to \RR) \to \RR \]
In a similar vein, the $\argmax$ operator is also a higher-order function of a particular type: 
\[ \argmax \colon (X \to \RR) \to \power{X} \]
where $\power{X}$ is the set of subsets of $X$. 
For a given function $u \colon X \to \RR$ we have that $\argmax(u)$ is the set of points where $u$ attains its maximum value. 

\subsection {Agent Context}

We want to model an agent $\mathcal A$ in an economic \emph{situation} or \emph{context} and formulate his motivations and his choices. We shall model such
contexts as mappings $X \to R$ that encode for choices in $X$ their effects on the outcomes in $R$.

\begin{definition}[Agent context] 
We call any function $p : X \to R$ a possible \emph{context} for the agent $\mathcal A$ who
is choosing a move from a set $X$, having in sight a final outcome in a set $R$, 
\end{definition}

For instance, $X$ could be the set of available flights between two cities, and $R = \RR^+$ could be the set of positive real numbers that represent prices. An agent who is interested in choosing a flight having in mind only the cost of the flight will consider the price list $X \to \RR^+$ as a sufficient context for his decision. If, however, the number of stops (or changes) is an important factor in the decision of the agent, we could take $R = \RR^+ \times \NN$ and the agent's context would then be $X \to \RR^+ \times \NN$. 

\subsection {Quantifiers}
\label{sec:quantifier_preference_util}

Suppose the agent $\mathcal A$ has to make a decision in the context $p \colon X \to R$. The agent will consider some of the possible outcomes to be \emph{good} (or \emph{acceptable}), and others to be bad (or \emph{unacceptable}). Such choices define a higher-order function of the following type:

\begin{definition}[Quantifier, \cite{escardo10a,escardo_sequential_2011}] Mappings $$\varphi : (X \to R) \to \power{R}$$ from contexts $p : X \to R$ to sets of outcomes $\varphi(p) \subseteq R$ are called \emph{quantifiers}.\footnote{The terminology comes from the observation that the usual existential $\exists$ and universal $\forall$ quantifiers of logic can be seen as operations of type $(X \to \BB) \to \BB$, where $\BB$ is the type of booleans. Mostowski \cite{Mostowski_1957} has called arbitrary functionals of type $(X \to \BB) \to \BB$ \emph{generalised quantifiers}. This was generalised further  in \cite{escardo10a} to the type given here.}
\end{definition}

We model agents $\mathcal A$ as quantifiers $\varphi_{\mathcal A}$ and take $\varphi_{\mathcal A}(p)$ as the set of outcomes that the agent $\mathcal A$ considers preferable in each context $p \colon X \to R$. Our main objective in this paper is to convince the reader that this is a general, modular, and highly flexible way of describing an agent's goal or objective.

The classical example of a quantifier is utility maximisation. Suppose an agent has a utility function $u'  \colon R  \to \RR$ mapping outcomes into utilities. Composing the context $p \colon X \to R$ and $u' \colon R \to \RR$ we get a new context that maps actions directly into utility $u \colon X \to \RR$. Given this new context, the good outcomes for the player are precisely those for which his utility function is maximal. This quantifier is given by
\[ \max (u)= \{ r \in \Img(u)\mid r \geq u(x') \text{ for all } x' \in X \} \]
where $\Img(u)$ denotes the image of the utility function $u \colon X \to R$. 

\subsection {Context-dependence}
\label{sec:context-dependence}

In general, we are going to allow the set of outcomes that the agent considers good to be totally arbitrary. It is reasonable, however, to assume that for each context $p \colon X \to R$ we have $\varphi(p) \neq \varnothing$. This is to say that in any given context the agent must have a preferred outcome (even if this would be the least bad one). We will call such quantifiers \emph{total}. Another more interesting class of quantifiers consists of those we call \emph{context-independent}:

\begin{definition}[Context-independence] A quantifier $\varphi \colon (X \to R) \to \power{R}$ is said to be \emph{context-independent} if the value $\varphi(p)$ only dependents on $\Img(p)$, i.e.
\[ \Img(p) = \Img(p') \implies \varphi(p) = \varphi(p'). \]
Dually, a quantifier $\varphi$ will be called \emph{context-dependent} if for some contexts $p$ and $p'$, with $\Img(p) = \Img(p')$, the sets of preferred outcomes $\varphi(p)$ and  $\varphi(p')$ are different.
\end{definition} 

Intuitively, a context-dependent quantifier will select good outcomes not only based on which outcomes are possible, but will also take into account how the outcomes are actually achieved. It is easy to see that the quantifier $\max$ is context-independent, since it can be written as a function of $\Img (p)$ only.

Our prototypical example of a context-dependent quantifier is the fixpoint operator
\[ \fix : (X \to X) \to \mathcal P (X) \]
Recall that a fixpoint of a function $f : X \to X$ is a point $x \in X$ satisfying $f(x) = x$. If the set of moves is equal to the set of outcomes then there is a quantifier whose good outcomes are precisely the fixpoints of the context. If the context has no fixpoints we shall assume that the agent will be equally satisfied with any outcome. 
Such a quantifier is given by
\[ \fix (p)= \begin{cases}
\{ x \in X \mid p(x) = x \} &\text{ if nonempty } \\
X &\text{ otherwise}.
\end{cases} \]
Clearly $\fix(\cdot)$ is context-dependent, since we could have different contexts $p, p' \colon X \to X$ having the same image set $\Img(p) = \Img(p')$ but with $p$ and $p'$ having different sets of fixpoints. For example, if we take $p, p' : \RR \to \RR$ to be given by $p (x) = x$ and $p' (x) = -x$ then $\Img (p) = \Img (p') = \RR$, but $\fix (p) = \RR$ and $\fix (p') = \{ 0 \}$. We will discuss the economic relevance of this particular quantifier in Section \ref{sec:reflexive} where we discuss reflexive agents.

\subsection{Attainability}

Another important property of quantifiers that we shall consider is that of \emph{attainability}:

\begin{definition}[Attainability] \label{def-attain} A quantifier $\varphi : (X \to R) \to \mathcal P (R)$ is called attainable if, for every context $p : X \to R$, for some $r \in \varphi(p)$ there exists an $x$ such that $p(x) = r$. (In particular, attainable quantifiers are total.)
\end{definition}

In other words, an agent modelled by an attainable quantifier will select at least one preferred outcome $r$ that is actually \emph{achievable} by some move $x$. An equivalent definition is that $\varphi \colon (X \to R) \to \power{R}$ is attainable if and only if
\[ \varphi(p) \cap \Img(p) \neq \emptyset. \]

\begin{remark} We could also define a \emph{strong attainability} notion whereby all $r \in \varphi(p)$ need to be achievable by some $x \in X$, i.e.
\[ \varphi(p) \subseteq \Img(p). \]
For our purposes the weaker notion of Definition \ref{def-attain} has been sufficient and reasonably well-behaved.
\end{remark}

Attainable quantifiers bring out the relevance of moves in the decision making process. Sometimes an agent might actually wish to spell out the preferred \emph{moves} instead of the preferred \emph{outcomes}. This leads to the definition of another class of higher-order functions:

\begin{definition}[Selection functions] A \emph{selection function}\footnote{In the computer science literature where selection functions have been considered previously \cite{escardo10a,escardo_sequential_2011} the focus was on single-valued ones. However, as multi-valued selection functions are extremely important in our examples we have adapted the definitions accordingly.} is any function of type
\[ \varepsilon : (X \to R) \to \mathcal P (X) \]
\end{definition}

Similarly to quantifiers, the canonical example of a selection function is miximising $\RR$, defined by
\[ \argmax(p)= \{ x \in X \mid p(x) \geq p(x') \text{ for all } x' \in X \} \]
The $\argmax$ selection function is naturally multi-valued: a function may attain its maximum value at several different points.

\begin{proposition} A quantifier $\varphi \colon (X \to R) \to \power{R}$ is \emph{attainable} if and only if there exists a total selection function $\varepsilon \colon (X \to R) \to \power{X}$ such that, for all $p \colon X \to R$, 
\[ x \in \varepsilon (p) \implies p (x) \in \varphi (p)\]
\end{proposition}

If such a relationship between a quantifier $\varphi : (X \to R) \to \mathcal P (R)$ and a selection function $\varepsilon : (X \to R) \to \mathcal P (X)$ holds 
then we shall say that $\varepsilon$ \emph{attains} $\varphi$. The attainability relation holds between the quantifier $\max$ and the selection function $\argmax$. The fixpoint quantifier is also a selection function, and it attains itself since
\[ x \in \fix (p)\implies p (x) \in \fix(p). \]

\section{Utility Maximisation and Preference Relations}
\label{sec:decisions}

In this section we relate the concepts of quantifiers and selection functions to the standard concepts of classical decision theory: utility functions and preference relations. In particular, we show that both correspond to context-independent quantifiers that have the same structure. 
We now want to characterise the relationship between preference relations and context-independent quantifiers.

\subsection{Preference Relations and Context-Independent Quantifiers}

Suppose $R$ is the set of possible outcomes, and an agent has a partial order relation $\succeq$ on $R$ as preferences, so that $x \succeq y$ means that the agent prefers the outcome $x$ to $y$. These partial orders lead to choice functions $f : \mathcal P (R) \to \mathcal P (R)$ where $f(S)$ are the maximal elements in the set of possible outcomes $S$ with respect to the order $\succeq$. Note that these $f$ satisfy $f (S) \subseteq S$, and $f (S) \neq \varnothing$ for non-empty $S$.

Every such $f$ can be turned into a quantifier $\varphi$ in a generic way, using the fact that the image operator is a higher-order function $\Img : (X \to R) \to \mathcal P (R)$:
\[ (X \to R) \xrightarrow{\Img} \mathcal P (R) \xrightarrow{f} \mathcal P (R) \]
so that $f \circ \Img \colon (X \to R) \to \mathcal P (R)$ are quantifiers. 

\begin{proposition} Assume $|X| \geq |R|$, such that the number of choices is bigger than the number of possible outcomes. Then a quantifier $\varphi \colon (X \to R) \to \power{R}$ is context-independent if and only if $\varphi = f \circ \Img$, for some choice function $f \colon \power{R} \to \power{R}$.
\end{proposition}
\begin{proof} If $\varphi = f \circ \Img$ then $\varphi$ is context-independent. For the other direction, note that since $|X| \geq |R|$ we have for any subset $S \subseteq R$ a map $u_S \colon X \to R$ such that $\Img(u_S) = S$. Assume $\varphi$ is context-independent and let us define $f(S) = \varphi(u_S)$. Clearly,
\[ \varphi(p) = \varphi(u_{\Img(p)}) = f(\Img(p)) \]
where the first step uses that $\varphi$ is context-independent and that $\Img(p) = \Img(u_{\Img(p)})$ by the assumption on the family of maps $u_S$, while the second steps simply uses the definition of $f$.
\end{proof}

Agents who are defined by context-independent quantifiers are choosing the set of good outcomes simply by ranking the set of outcomes that can be achieved in a given context but are ignoring all the information about how each of the outcomes arise from particular choices of moves. 

For instance, we might have a set of actions that will lead us to earn some large sums of money. Some of these, however, might be illicit. A classical agent who cares only about the direct consequences of his decision and is defined in a context-independent way would choose the outcome that gives himself the maximum sum of money, regardless of the nature of action. If however the agent also cares about the actions themselves and their indirect consequences, he might not consider the largest amount of money as preferable. As outlined in the introduction a standard remedy in order to include such implicit concerns is to extend the outcome space. This can be a necessary correction by the analyst if the initial outcome space was truly misconceived. Encoding procedural concerns by redefining the outcome space may, however, come at a cost. We will come back to this methodological issue in Section \ref{sec:beyond}. 

We close this subsection with the following proposition which guarantees the attainability of context independent quantifiers arising from preference relations: 

\begin{proposition} Whenever $f_{\succeq}$ is a choice function arising from a partial order $\succeq$, then the context-independent quantifier $\varphi = f_{\succeq} \circ \Img$ is attainable.
\end{proposition}
\begin{proof} By the definition of $\varphi$ we have that if $r \in \varphi(p)$ then $r$ is a maximal element in $\Img(p)$. Hence we must have an $x \in X$ be such that $p(x) = r$.
\end{proof}

\subsection{Rational Preferences and Utility Functions}

The usual approach to model behaviour in economics is to either postulate a preference relation on the set of alternatives, as discussed above, or to directly assume a utility function \cite{kreps2012microeconomic}. A certain structure is imposed on preference relations mainly for two reasons: either because the additional structure is deemed to be a characteristic of an agent's rationality\footnote{This issue has been intensely debated, see
  \cite{Richter1971,Mas-Colell1995,kreps2012microeconomic}.}, or because one wishes to work with utility functions. It is a standard result that for utility functions to exist, preferences relations have to be \emph{rational} \cite{kreps2012microeconomic}.

Now, rational preferences and utility functions are special cases of the generic construction of a context-independent quantifier that we have outlined in the last section. Rational preferences are special because (i) we impose additional structure on $R$, that is, $\succeq$ is a total preorder and (ii) we focus on one particular $f_\succeq$, that is, $f_\succeq : \mathcal P (R) \to \mathcal P (R)$ defined by
\[ f_\succeq (S) = \{ \succeq \text{-maximal elements of } S \} \] 

A rational preference relation can always be represented by a utility function. Translated into our higher-order approach, the utility function can be characterised as the context $p \colon X \to \RR$ that attaches a real number to each element of the set of choices $X$ with the quantifier defined as
\[
  \varphi (p) = \max_{x \in X} p(x).
\]
Moreover, this quantifier is attained by the selection function
\[
  \varepsilon (p) = \arg\max p
\]
Note the types $\varphi \colon (X \rightarrow \RR) \rightarrow \mathcal P (\RR)$ and $\varepsilon \colon (X \rightarrow \RR) \rightarrow \power{X}$ respectively. And indeed we have that
\[ x \in \varepsilon (p) \implies p(x) \in \varphi (p). \]
Thus, $\max$ and $\arg \max$ operators, which are universally used in the economic literature, are the prototypical examples of a context-independent quantifier and a selection function attaining it. Since utility functions and preference relations can be both represented by quantifiers we can conveniently work with both representations in one model if represented in our framework. 

As long as a rational preference relation is a good representation of a decision problem, there is no obvious reason why not to use utility functions. In fact, often utility functions are seen as more convenient because 
of the availability of standard optimisation techniques. They also provide a succinct description of the agent's goals \cite{Rubinstein2006}. 

However, there are reasons why alternatives to maximisation are important. Firstly, there are situations where utility functions are not applicable. A canonical example is the case of lexicographic preferences. Secondly, in the next section, we will show that higher-order functions also represent decision procedures based on 
other than classical utility functions and preference relations. 
We will discuss why economists should care about having such alternative representations at hand.

\section{Alternatives to Optimisation}
\label{sec:beyond}

We have seen how the higher-order notion of a context-independent quantifiers is able to model choices based on rational preferences (or equivalently on utility maximisation). In this section we show that we can go beyond these cases by allowing for a different structure on the set of outcomes $R$ or by allowing for a different mapping $f \colon \power{R} \to \power{R}$, or by relaxing both.

Firstly, we show that quantifiers include decision procedures that cannot be easily modelled by rational preference relations. Secondly, we will argue that even if it were possible to model a decision problem by rational preferences or utility maximisation, it may be insightful to have alternative representations at hand. 

Why should economists be interested in modelling behaviour differently, if a representation of utility functions is possible? The main methodological question is whether all decisions shall be modelled as being motivated by their consequences. It is clear that in many situations it is possible to redefine the outcome space such that procedural aspects or menu-dependence are encoded in the outcome space. So why should economists be interested in alternatives at all? 

\subsection{Beyond Rational Preferences }

The assumption that the preorder on the outcome space is total, which guarantees the existence of a utility function, is demanding and in fact more demanding than is necessary to rationalise choice behaviour \cite{Richter1971}. When taking the perspective of preferences, from a positive as well as a normative viewpoint, there are good reasons why a rational decision-maker may exhibit indecisiveness, meaning that his preference for some pairs of outcomes may not be defined \cite{Aumann1962}. 
Moreover, consider a situation where the economist or some other agent has only partial information about the preferences of an agent and regards him ``as if'' he has incomplete preferences \cite{Dubra_et_al2004_comple}.
Lastly, $R$ may be a set of alternatives to be chosen by a group of agents. Even if each individual's preferences are complete, the aggregate social welfare ordering does not have to be \cite{Ok2002}.

Sen \cite{Sen1997} discusses ``inescapability or urgency of choice" and situations where the agent has to decide even if he has not totally ranked all alternatives. As a result, no optimal choice in a classical sense can be made. However, Sen \cite{Sen1997} claims that completeness is not necessarily a condition for maximisation which only requires that our chosen alternative is not known to be worse than any other. 

There have been various attempts to change standard formalisms to allow for a utility theory without the need to fulfil the completeness assumption.\footnote{For an important early contribution see \cite{Aumann1962}. More recent contributions include \cite{Ok2002} for utility representations in certain environments and \cite{Dubra_et_al2004_comple} for uncertain environments. See also references in \cite{Ok2002}.} When working with quantifiers and selection functions, the set of outcomes $R$ can have \textit{any} order, or no underlying structure at all. In particular, the preference relation does not have to be total. That is, given any preference relation $\succeq \, \subseteq R \times R$, an agent chooses the best alternatives as outlined above. By that one can consider choices that are not in the scope of utility functions without the need to change the framework. To be clear, the selection function that corresponds to the preference ordering $\succeq$ is
\[ \argmax (p) = \{ x \in X \mid r \succeq p(x) \implies r \not\in \operatorname{Im} (p) \} \]
i.e. a maximal outcome is one which is not known to be worse than any attainable outcome.
It is important to notice that we may still have a total quantifier, even if the preference relation is not total. 
Total quantifiers guarantee the existence of a preferred outcome even in a situation of incomplete preferences. 

\subsection{Beyond Utility Functions}

The utility approach is intimately linked to the assumption that the agent fully optimises. The behavioural economic literature as well as the psychological literature have documented deviations from optimising behaviour \cite{camerer2011behavioral,kahneman2011thinking}. Quantifiers provide a direct way to model such deviations. Here we give a few examples how to represent these cases in our framework.

\begin{example}[Averaging Agent] Consider an agent who prefers the outcome to be as close as possible to the average of all achievable outcomes. Given a decision context $p \colon X \to \RR$, the average amongst the possible outcomes can be calculated as
\[ A_p = \frac{\Sigma_{r \in \Img(p)} r}{|\Img(p)|} \]
Therefore, such agent can be directly modelled via the averaging quantifier $\varphi^A \colon (X \to \RR) \to \power{\RR}$ as
\[ \varphi^A(p) = \{ r \in \Img(p) \;\mid\; \mbox{$|r - A_p|$ is minimal} \} \]
\end{example}

\begin{example}[Ideal-move Agent \cite{hedges13}] Let $r > 0$ be a fixed real number. For a point $v \in \RR^n$ we define the closed ball with centre $v$ and radius $r$ by
\[ B (v; r) = \{ w \in \RR^n \mid d (v, w) \leq r \} \]
where $d$ is the Euclidean distance. Let the set of choices $X$ have a distinguished element $x_0 \in X$. Define the quantifier $\varphi : (X \to \RR^n) \to \mathcal P (\RR^n)$ by
\[ \varphi (p) = B (p(x_0); r) \]
This quantifier is attained by the constant selection function $\varepsilon (p) = \{x_0\}$. 
\end{example}

The last example illustrates Simon's satisficing behaviour. The value $r > 0$ can be considered as a satisficing threshold around outcomes that are close to the outcome of an ideal point. Such an agent is equally satisfied with all outcomes which are close enough to the outcome of the ideal choice. 

The next example represents the second best decision problem discussed in \cite{Kalai2002}.

\begin{example}[Second-best Agent] Consider a simple heuristic of a person ordering wine in a restaurant whereby he always chooses the second most-expensive wine. In terms of quantifiers, let $X$ be the set of wines available in a restaurant, and $p: X \rightarrow \RR$ the price attached to each wine $x_i$ ($i=1, ...,N$) on the menu, so that $r_i = p(x_i)$ denotes the price of wine $x_i$. Given a maximal strict chain $r_n > r_{n-1} > \ldots > r_1$ in $\RR$, let us call $r_{n-1}$ a sub-maximal element. The goal of the agent can be described by the quantifier
\[ \varphi_{>}(p^{X \to \RR}) = \{ \mbox{sub-maximal elements with respect to $>$ within $\Img(p)$} \}. \]
\end{example}

A crucial point of the above examples is the additional degree of freedom of modelling as it is possible to vary the choice operator itself and not being automatically restricted to the max operator and to consider behaviour to be necessarily rationalised by rational preferences.

Obviously, one could rationalise the above choices as the outcome of a maximisation. One could redefine preferences and utility functions such that the outcome of the maximisation is just the second most expensive wine.\footnote{Note, if the prices of the wines represented preferences, a rationalisation of second best choices is not possible (see \cite{Rubinstein2006}).}  However, while equivalent in outcome, the causal model of behaviour is different. The classical approach would force the choice to be rational, whereas in our setting this question remains open. The quantifier formally describes an agent's behaviour. It could be that the choice pattern is a habitual heuristic or it could be the reduced form pattern of some rational decision making in a larger context.

Of course, instead of choosing the second most expensive wine, one could consider alternative heuristics, such as choosing the wine that is closest to the average price of all available wines on the menu, or within a class of wines, etc.

Moreover, one could also combine any heuristic with some arbitrary preferences. Say, the guest is a fan of white wines, and he strictly prefers Chardonnay over Riesling. One could model the agent as first restricting the choices to the wines that are Chardonnay (if available) and then apply his second most expensive decision heuristic to the class of Chardonnay available.

\subsection{Context-Dependent Decision Problems}\label{sec:context_dec}

So far, we have focused only on context-independent quantifiers. 
We have seen that already this restricted class of quantifiers can take us beyond choices motivated by rational preferences. 
Yet, we can do more. As we have discussed in Section 2, we can allow for quantifiers that do not only take the image of $p$ as input but the complete function. Again, we consider several examples.

\begin{example}[Averaging -- revised] Consider again an agent who prefers the outcome to be as close as possible to the average outcome.
But this time we assume that he takes into account the number of possible ways an outcome may be attained. Given a decision context $p \colon X \to \RR$, the weighted average in this case can be calculated as
\[ A_p = \frac{\Sigma_{x \in X} p(x)}{|X|} \]
Such agent can be modelled via the weighted averaging quantifier $\varphi \colon (X \to \RR) \to \power{\RR}$ as
\[ \varphi(p) = \{ r \in \Img(p) \;\mid\; \mbox{$|r - A_p|$ is minimal} \} \]
It easy to check that this is a \emph{context-dependent} quantifier.
\end{example}

Now, consider the example where the set of actions allows an agent to earn some money but some actions are illicit and hence not considered to be a permissible behaviour.  If we care about the actions themselves, we might not necessarily consider the largest sum of money as preferable. This example corresponds to the discussion in (\cite{Sen1997}, section 6) and goes back to Adam Smith (as quoted in \cite{Sen1997}): ``\emph{\ldots conduct that go beyond the pursuit of specified goals has a long tradition. As Adam Smith (1790) had noted, our behavioural choices often reflect `general rules' that `actions' of a particular sort `are to be avoided' (p. 159)}". 

Sen \cite{Sen1997} proposes two methods in order to represent such situations, where the first method is the one that is broadly used in the economic literature:

(1) Incorporate the context or concerns about actions explicitly into an extended outcome space by rewriting the set of outcomes: all outcomes have to be ranked by hand and payoffs are assigned accordingly.  For example, this can be done by attaching the appropriate negative values to all monetary outcomes which are achieved by a criminal activity (note that we do not mean the legal costs of criminal activities which can be easily monetised).  The agent then behaves ``as if" he is maximising this new set.

(2) Restrict the choice options further by taking a permissible subset of actions, reflecting self-imposed constraints or social norms of permissible behaviour, and then seek the maximal outcome from the set of achievable outcomes. 

The optimal actions in the examples in \cite{Sen1997}  are the same if modelled by these two approaches.
There is, however, some critique of these methods. Regarding ``as if" preferences Sen \cite{Sen1997} notes that this new set is ``\emph{\ldots a devised construction and need not have any intuitive plausibility seen as preference. A morally exacting choice constraint can lead to an outcome that the person does not, in any sense, `desire', but which simply mimics the effect of his self-restraining constraint... The `as if' preference works well enough formally, but the sociology of the phenomenon calls for something more than the establishment of formal equivalences.}" 

Such new sets do not represent the original goals of the agents. It may be transparent in a simple example, but not so in situations involving several agents with different goals engaged in interactions. Moreover, we have to derive such a set for each problem separately and cannot represent the original goal as a general rule of conduct, that is combinable with some other concerns, like: ``in any arbitrary situation, first consider the socially permissible actions and then maximise over outcomes which follows from permissible actions". An additional complication arises if the preferences over actions are not just given, but are the product of some more general process, for example a game. Then we have to rewrite the whole ``as if" set each time the social norms change. 

The restriction function seems to represent the choice over the actions in a more intuitive way, but only for some specific examples where some actions are just thrown away from the consideration. It cannot represent more complicated preferences over the action set. In our framework we can consider arbitrary combinations of preferences over actions and outcomes. 
For example, our agent can have a lexicographic ordering of permissible and non-permissible actions or, alternatively, may pursue socially non-permissible activities if they are not too profitable, but may consider them 
if the resulting outcome is high enough. 

\begin{example}[Honest Agents] Consider an agent with a set of possible actions $X$ leading to monetary outcomes $M \subseteq \RR$. Assume some of these actions $I \subset X$ are illegal or dishonest. Hence, the set $L = X \backslash I$ consists of the legal, or honest, actions. In the first instance consider an honest agent who maximises over the outcomes which follows from honest actions. Such a honest agent can be modelled by the quantifier: 
\begin{equation*}
\varphi^h (p) = \{ r \; \mid \; \mbox{$r$ a maximal element in the set $p(L)$} \}
\end{equation*}
where $p(L)$ is the image of $L$ under $p$. Consider, however, a more complicated case where the agent is prepared to consider dishonest or illegal actions when the reward associated with some of these actions is above a threshold $T$. This subtler preference can be directly modelled as
\[
\varphi^d(p) = 
\begin{cases}
	\{ r \; \mid \; \mbox{$r$ is maximal in $\Img(p)$} \} & \text{ if } \max_{x \in I} p(x) > T \\[1mm]
	\varphi^h(p) &\text{ otherwise }
\end{cases}
\]
so that the dishonest agent will behave as the honest one if the maximal reward for a dishonest action is low, but he will consider any action to be acceptable if the gain from a dishonest or illegal action is high enough.
\end{example}

In the next example we introduce an extreme case of an agent who decides on preferred outcomes solely based on the set of moves that lead to that outcome. 

\begin{example}[Safe Agents] Given a decision context $p \colon X \to R$ and an outcome $r \in \Img(p)$, we can calculate the number of different ways $r$ can be attained by
\[ n^p_r = \left| \{ x \in X \mid p(x) = r \} \right|. \]
We say that an outcome $r$ is most unavoidable if $n^p_r$ is maximal over the set of possible outcomes $\Img(p)$. We say that an agent is safe if he prefers most unavoidable outcomes. Such agents are modelled by the quantifier
\[
\phi(p) = \{ r \in \Img(p) \; \mid \; \mbox{ $n^p_r$ maximal} \; \} 
\]
\end{example}

In order to illustrate this quantifier, suppose there are three beaches, and the agent is indifferent between them. The first can be reached by one highway, the second by two highways and the third by three highways. The agent has to choose which highway to take, and the outcome is the beach that the agent goes to. The safe agent decides to visit the beach which can be reached by the most different routes, which is the third, in order to avoid the risk of being stuck in a traffic jam.

\section{Reflexive Agents}
\label{sec:reflexive}

We now discuss the specific situations where the set of actions $X$ and outcomes $R$ are the same $X = R$. In this case elements of the type
\[ 
(X \to X) \to \power{X}
\]
can be either viewed as quantifiers or selection functions. 
Agents of this type are common in elections:

\begin{example} [Voting Agent] \label{keynesian-ex1} Consider three judges $J=\{J_1, J_2, J_3\}$ voting for two contestants $X = \{A, B \}$. The winner is determined by the simple majority rule of type $\operatorname{maj}: X\times X\times X \rightarrow X$. The set $X$ denotes both the set of choices and the set of possible outcomes of the contest. We first assume that the judges rank the contestants according to a preference ordering. For example, suppose judges 1 and 2 prefer $A$ and judge 3 prefers $B$. Consider the decision problem of the first judge. He has an ordering on the set $X$, namely $A \succeq_1 B$, and his goal is to maximise the outcome with respect to this ordering. Hence, he is modelled via the quantifier: 
\[
\varphi_1^J(p) = \max_{x \in (X, \succeq_1)} p(x)
\]
\end{example}

The set $X$ is equipped with a partial order and the $\max$ operator $(X \to X) \to \power{X}$ describes the agent. 

Another very interesting example of an agent with an important economic interpretation, is the fixpoint operator, that we have already 
mentioned in Section \ref{sec:context-dependence}.

\begin{example}[Keynesian Agent] \label{keynesian-ex2} Consider the same example as in the last example but now assume that judge 1 has different preferences: he prefers to support the winner of the contest. He is only interested in voting for the winner of the contest and he has no preferences for the contestants per se.  The selection function of such a Keynesian agent can be described by a fixpoint operator as
\[
\varepsilon_1^K(p) = \fix(p) = \{ x \in X \; \mid \; p(x) = x \}.
\]
Interestingly, such an agent is best described by a selection function, rather than via the corresponding quantifier
\[
\varphi_1^K(p) = \{ p(x) \; \mid \; p(x) = x \}.
\]
\end{example}

We note that it is perfectly possible to model such a Keynesian agent via standard utility functions, attaching say utility 1 to good outcomes and 0 to the bad ones, so that the judges maximise over the set of monetary payoffs. In this process of attaching utilities to the decision, however, one has to compute the outcome of the votes, 
then to check for the second and the third judges whether their vote is in line with the outcome, and finally to attach the utilities. In some sense, the economist takes the whole decision process, solves the problem, identifies the good outcomes according to the natural language description of the problem, and then lets the agent to choose between 1 and 0. 

On the other hand, if we use the fixpoint operator in order to represent the goal, we equip the individual agent himself with the problem solving ability 
that we as the modeller otherwise use in order to compute the utility such that the utility maximising agent behaves as if he were a Keynesian agent.
We have been therefore tempted to call such fixpoint agents \emph{reflexive agents} 
as they do inside the model what the economist is doing outside the model.
These fixpoint agents with their computational power resemble a construction that is at the core of the 
Lucas critique.\footnote{Sargent \cite{Sargent1993} describes the need for a similarity of the economist and the economically reasoning agents 
in the economists' models as follows:
{\it  ``[t]he idea of rational expectations is ... said to embody the idea that economists 
and the agents they are modeling should be placed on the  equal footing: 
the agents in the model should be able to forecast and profit-maximize 
and utility-maximize as well as the economist - or should we say the econometrician - who constructed the model."}}

As briefly discussed above, most functions $p \colon X \to X$ do not have a fixpoint and the fixpoint operator will often give the empty set. For the purposes of modelling a particular situation we might want to totalise the fixpoint operator in different ways and describe what an agent might do in case that no fixpoint exists. 
The fixpoint goals are far more interesting when we consider a game with several agents with different concerns, for instance some with usual preferences 
and some with fixpoint goals.
We analyse such a game in detail in our companion paper on higher-order games \cite{Hedges_et_al_2015_games}.

Let us conclude with another example of a reflexive agent.

\begin{example}[Coordinating Agent] \label{keynesian-ex3} Consider two players, $\{0,1\}$, who want to coordinate, for instance, 
about the restaurant where to meet for lunch. The set of actions $X_0 = X_1 =\{A, B\}$ denotes the different restaurants at choice.
The set of outcomes $R = X_0 \times X_1$ denotes the two restaurants where the agents might end up. 
The fact that these two agents want to meet in the same restaurant ca be directly described by another sort of fixpoint operator:
\[
\varepsilon_i(p) = \{ x \in X_i \; \mid \; x = (\pi_{1 - i} \circ p)(x) \}
\]
where $\pi_i \colon X_0 \times X_1 \to X_i$ are the projection functions. The preferred move of agent $i$ is the one which leads him to the same place as the other agent $1 - i$.
\end{example}

These two examples above show that the overall goal of the Keynesian and the coordinating agent are very similar, and can be captured by some variants of fixpoint operators. Even though it is possible to use utility functions in order to model these concerns in the particular examples, it is not so obvious that this commonality can be made explicit when modelling with utility functions. In our more abstract formalisation via higher-order functions, it is possible to detect patterns across problems that are hard to find when one only looks at the compiled level of utility maximisation. 

\section{Conclusions}
\label{sec:conclusions}

The utility maximisation framework of standard decision theory is an intuitive representation of rational agents and the yardstick approach in economics. 
The rationale for adhering to this approach is that modelling tools, 
such as functional analysis as well as solution methods, such as optimisation via Lagrangian methods, are readily available.
However, theory, reality and experiments suggest that we need to go beyond the modelling strategy of representing any behaviour as if it were the result of utility maximisation. Moreover, it is not straightforward how to implement utility maximisation in computers or take into account computability issues for example for the real numbers of utility maximisation.

Our higher-order approach offers a path to resolve both of these issues: we can generalise utility maximisation.
At the same time higher-order functions provide a way to implement our games in computers. 
Higher-order functions generalise the $\max$ and $\argmax$ operators, and as the foundation of computability theory, programming language and compiler design 
and computer science in general, they form a way to implement models in computers and analyse decision models with the tools to analyse software.

Another, very important issue of the higher-order functions is that they equip our decision framework with the feature of programming languages being compositional.
In this paper we have seen how various decision goals are composable, in our companion paper \cite{Hedges_et_al_2015_games} we show that this extends to games as well.
Games are composable from decisions and furthermore algebraically composable into any complicated game.

\bibliographystyle{plain}
\bibliography{../references}



\end{document}